\documentclass[a4paper,11pt]{scrartcl}

\usepackage[utf8]{inputenc}
\usepackage{lmodern}

\usepackage{paralist}

\newcommand{\nor}[1]{\left\|#1\right\|}

\newcommand{\ZZ}{\mathbb{Z}}
\newcommand{\ZZgeq}{\mathbb{Z}_{\geq 0}}
\newcommand{\QQ}{\mathbb{R}}
\newcommand{\QQgeq}{\mathbb{R}_{\geq 0}}
\newcommand{\pp}{b}

\newcommand{\Ptop}{\mathcal{P}}

\newcommand{\PP}{\Pi}
\newcommand{\XX}{X}

\usepackage{amsmath}
\usepackage{amsthm}
\usepackage{thmtools}
\usepackage{amsfonts}
\usepackage{amssymb}
\usepackage{graphicx}
\usepackage{tikz}
\usepackage{nicefrac}
\usepackage{ifthen}
\usepackage{caption}
\usepackage{subcaption}
\usepackage[left=1in,right=1in,top=1in,bottom=1in]{geometry}

\usepackage{thm-restate}

\usepackage[style=alphabetic,backend=bibtex,safeinputenc,maxnames=6,block=space]{biblatex}
\usetikzlibrary{shapes,backgrounds,fit,patterns, decorations.pathmorphing}

\usepackage{hyperref}
\usepackage{csquotes}
\usepackage[final]{microtype}

\usetikzlibrary{decorations.pathreplacing, positioning, arrows, calc}

\theoremstyle{plain}
\newtheorem*{theorem*}{Theorem}
\newtheorem*{lemma*}{Lemma}
\newtheorem{theorem}{Theorem}
\newtheorem{corollary}{Corollary}
\newtheorem{lemma}{Lemma}
\newtheorem*{claim}{Claim}
\newtheorem*{case}{Case}
\newtheorem{observation}{Observation}
\newtheorem{definition}{Definition}
\newtheorem{mainthm*}{Main Theorem}

\usepackage[printonlyused]{acronym}

\acrodef{afptas}[AFPTAS]{asymptotic fully polynomial time approximation
  scheme}
  \acrodef{aptas}[APTAS]{asymptotic polynomial time approximation
  scheme}
\acrodef{ilp}[ILP]{integer linear program}
\acrodef{lp}[LP]{linear program}
\acrodef{ptas}[PTAS]{polynomial time approximation scheme}
\acrodef{amfptas}[AMFPTAS]{asymptotic (fully) polynomial time approximation scheme}
  
\usepackage{authblk}

\title{About the Structure of the Integer Cone and its Application to Bin Packing\footnote{This work was partially supported by DFG Project, Entwicklung und Analyse von effizienten polynomiellen Approximationsschemata f\"ur Scheduling- und verwandte Optimierungsprobleme, Ja 612/14-2}}

\author[1]{Klaus Jansen}
\author[1]{Kim-Manuel Klein}
\affil[1]{Department of Computer Science, University of Kiel\\ \{kj,kmk\}@informatik.uni-kiel.de}

\date{}

\begin{document}
\maketitle
\begin{abstract}
    We consider the bin packing problem with $d$ different item sizes and revisit the structure theorem given by Goemans and Rothvo\ss~\cite{goemans2015} about solutions of the integer cone. We present new techniques on how solutions can be modified and give a new structure theorem that relies on the set of vertices of the underlying integer polytope. As a result of our new structure theorem, we obtain an algorithm for the bin packing problem with running time $|V|^{2^{O(d)}} \cdot enc(I)^{O(1)}$, where $V$ is the set of vertices of the integer knapsack polytope and $enc(I)$ is the encoding length of the bin packing instance. The algorithm is fixed parameter tractable, parameterized by the number of vertices of the integer knapsack polytope $|V|$. This shows that the bin packing problem can be solved efficiently when the underlying integer knapsack polytope has an easy structure, i.e. has a small number of vertices.
    
    Furthermore, we show that the presented bounds of the structure theorem are asymptotically tight. We give a construction of bin packing instances using new structural insights and classical number theoretical theorems which yield the desired lower bound.
\end{abstract}

\section{Introduction}
Given the polytope $\Ptop = \{ x \in \QQ^d \mid Ax \leq c\}$ for some matrix $A \in \ZZ^{m \times d}$ and a vector $c \in \ZZ^d$. We consider the integer cone
\begin{align*}
int.cone(\Ptop \cap \ZZ^d) = \{ \sum_{p \in \Ptop \cap \ZZ^d} \lambda_p p  \mid \lambda \in \ZZgeq^{\Ptop \cap \ZZ^d} \}
\end{align*}
of integral points inside the polytope $\Ptop$.
Let $\Ptop_I = Conv(\Ptop \cap \ZZ^d)$ be the convex hull of all integer points inside $\Ptop$, where for given set $X \subset \QQ^d$, the convex hull of $X$ is defined by $Conv(X) = \{\sum_{p \in X} x_p p \mid x \in [0,1]^X, \nor{x}_1 = 1 \}$.
Let $V_I$ be the vertices of the integer polytope $\Ptop_I$ i.e. $\Ptop_I = Conv(V_I)$. 
In case of the (fractional) cone $Cone(\Ptop \cap \ZZ^d) = \{ \sum_{p \in \Ptop \cap \ZZ^d} \lambda_p p  \mid \lambda \in \QQgeq^{\Ptop \cap \ZZ^d} \}$, we know by Caratheodory's Theorem (see e.g. \cite{schrijver86}) that each $\gamma \in \Ptop_I$ can be written as a convex combination of at most $d+1$ points in $V_I$ and hence $Cone(\Ptop \cap \ZZ^d) = Cone(V_I)$.

In this paper we investigate the structure of the integer cone $int.cone(\Ptop \cap \ZZ^d)$ versus $int.cone(V_I)$. Therefore, we define the \emph{vertex distance} of a point $b \in int.cone(\Ptop \cap \ZZ^d)$ which describes how many extra points from $(\Ptop \cap \ZZ^d) \setminus V_I$ are needed to represent $b$. The vertex distance is defined by
\begin{align*}
Dist(b) = \min \{ \nor{\gamma}_1 \mid \gamma \in \ZZgeq^{\Ptop \cap \ZZ^d}, \lambda \in \ZZgeq^{V_I} \text{ such that } b = \sum_{v \in V_I} \lambda_v v + \sum_{p \in \Ptop_I} \gamma_p p \}
\end{align*}
In this paper we show that for every point $b \in int.cone(\Ptop \cap \ZZ^d)$, the vertex distance $Dist(b)$ is bounded by $2^{2^{O(d)}}$. Hence, every $b$ can be written by $b = \sum_{v \in V_I} \lambda_v v + \sum_{p \in \Ptop \cap \ZZ^d} \gamma_p p$ for some $\lambda \in \ZZgeq^{V_I}$ and some $\gamma \in \ZZgeq^{\Ptop \cap \ZZ^d}$, where $\nor{\gamma}_1 \leq 2^{2^{O(d)}}$.

A related result concerning the structure of the integer cone was given by Eisenbrand and Shmonin~\cite{eisenbrand2006}. They proved that every $b \in int.cone(\Ptop \cap \ZZ^d)$ can be written by a vector $\lambda \in \ZZgeq^{\Ptop \cap \ZZ^d}$ with $b = \sum_{p \in \Ptop \cap \ZZ^d} \lambda_p p$ such that $\lambda$ has a bounded support (=number of non-zero components). Therefore, for a given set $M$ and given vector $\lambda \in \QQgeq^M$, let $supp(\lambda)$ be the set of non-zero components of $\lambda$, i.e. $supp(\lambda) = \{s \in M \mid x_s \neq 0 \}$.
\begin{theorem}[Eisenbrand, Shmonin~\cite{eisenbrand2006}] \label{thm-supp}
Given polytope $\Ptop \subset \QQ^d$. For any integral point $\pp \in int.cone(P \cap \ZZ^d)$, there exists an integral vector $\lambda \in \ZZgeq^{\Ptop \cap \ZZ^d}$ such that $\pp = \sum_{p \in \Ptop \cap \ZZ^d} \lambda_p p$ and $|supp(\lambda)| \leq 2^d$.
\end{theorem}

Let $(s,b)$ be an instance of the bin packing problem with item sizes $s_1, \ldots, s_d \in (0,1]$ and multiplicities $b \in \ZZgeq^d$ of the respective item sizes. The objective of the bin packing problem is to pack all items $b$ into as few unit sized bins as possible.
When we choose $\Ptop$ to be the knapsack polytope, i.e. $\Ptop = \{ x \in \ZZgeq^d \mid s^T x \leq 1 \}$, then a vector $\lambda \in \ZZgeq^{\Ptop_I}$ of $int.cone(\Ptop \cap \ZZ^d)$ yields a packing for the bin packing problem.
A long standing open question was, if the bin packing problem can be solved in polynomial time when the number of different item sizes $d$ is constant. This problem was recently solved by Goemans and Rothvo\ss~\cite{goemans2015} using similar structural properties of the integer cone. They proved the existence of a distinguished set $X \subset \Ptop$ of bounded size such that for every vector $\pp \in int.cone(\Ptop \cap \ZZ^d)$ there exists an integral vector $\lambda \in \ZZgeq^{\Ptop \cap \ZZ^d}$ where most of the weight lies in $X$. More precisely, they proved the following structure theorem:
\begin{theorem}[Goemans, Rothvo\ss~\cite{goemans2015}]
Let $\Ptop = \{ x \in \mathbb{R}^{d} \mid Ax \leq c \}$ be a polytope with $A \in \ZZ^{m \times d}, c \in \ZZ^{d}$ such that all coefficients are bounded by $\Delta$ in absolute value. Then there exists a set $X \subseteq \Ptop \cap \ZZ^{d}$ with $|X| \leq m^{d}d^{O(d)} (\log \Delta)^{d}$ such that for any point $\pp \in int.cone(\Ptop \cap \ZZ^d)$, there exists an integral vector $\lambda \in \ZZgeq^{\Ptop \cap \ZZ^d}$ such that $\pp = \sum_{p \in \Ptop \cap \ZZ^d} \lambda_p p$ and
\begin{enumerate}
\item $\lambda_p \leq 1 \qquad \forall p \in (\Ptop \cap \ZZ^d) \setminus X$
\item $|supp(\lambda) \cap X| \leq 2^{2d}$
\item $|supp(\lambda) \setminus X| \leq 2^{2d}$
\end{enumerate}
\end{theorem}
The set $X$ is constructed in \cite{goemans2015} by covering $\Ptop$ by a set of integral parallelepipedes. The set $X$ consists of the vertices of the integral parallelepiped and can be computed in a preprocessing step. Note that by the construction of Goemans and Rothvo\ss, we have that $V_I \subset X$ as the set of vertices of some inner centrally symmetric polytopes is computed.


\subsection{Our results:}
At first, we study the special case when $\Ptop$ is given by the convex hull of integral points $B_0, B_1, \ldots , B_d \in \ZZ^d$ i.e. $\Ptop$ is the simplex $S = Conv(B_0, B_1, \ldots , B_d)$. This is for example the case in the knapsack polytope when all items sizes are of the form $s_i = 1 / a_i$ for some $a_i \in \ZZ_{\geq 1}$. In this case, all vertices of the knapsack polytope are of the form $B_0 = (0, \ldots , 0)^T$ and $B_i = (0, \ldots ,0 , a_i , 0, \ldots ,0)^T$ for $1 \leq i \leq d$ and therefore integral. We prove the following theorem:
\begin{theorem} \label{thm-main}
Let $S$ be the simplex defined by $S = Conv(B_0, B_1, \ldots , B_d)$ for $B_i \in \ZZ^{d}$ and let $B$ be the set of vertices $B = \{B_0, B_1, \ldots , B_d \}$. For any vector $\pp \in int.cone(S \cap \ZZ^d)$, there exists an integral vector $\lambda \in \ZZ^{S \cap \ZZ^d}_{\geq 0}$ with $\pp = \sum_{s \in S \cap \ZZ^d} \lambda_s s$ and 
\begin{enumerate}
\item $\lambda_s \leq  2^{2^{O(d)}} \quad \forall s \in (S \cap \ZZ^d) \setminus B$
\item $|supp(\lambda) \setminus B| \leq 2^{d}$
\end{enumerate}
\end{theorem}
This theorem shows that in the case that the integer polytope $\Ptop_I$ is a simplex, the vertex distance $Dist(b)$ can be bounded by a term $2^{2^{O(d)}}$ for any $b \in int.cone(\Ptop \cap \ZZ^d)$.
In Section \ref{lower_bound}, we complement this result by giving a matching lower bound for $Dist(b)$. We prove that the double exponential bound for $Dist(b)$ is tight, even in the special case of bin packing, where the simplex $S$ is a specific knapsack polytope. The lower bound is based on the sylvester sequence $S_i$ which is inductively defined by $S_1 = 2$ and $S_{i+1} = (\prod_{j=1}^i S_j) + 1$ \cite{Graham1994}.
\begin{restatable}{theorem}{lowerbound}\label{thm-lower}
    There exists a bin packing instance with sizes $\frac{1}{a_1}, \ldots , \frac{1}{a_d}$ for $a_i \in \ZZ_{\geq 1}$ and multiplicities $b \in \ZZgeq^d$ corresponding to a point $b \in int.cone(\Ptop \cap \ZZ^d)$, where $\Ptop$ is the knapsack polytope such that 
    \begin{align*}
        Dist(b) \geq S_d-2 = 2^{2^{\Omega(d)}}
    \end{align*}
\end{restatable}
Furthermore, in the end of Section 4, we discuss the difficulty of finding instances with large vertex distance and we show a connection to the modified roundup property (see \cite{Scheithauer199793}).

As a direct consequence of our main Theorem \ref{thm-main}, we obtain a structure theorem that is similar to the one given by Goemans and Rotho\ss \cite{goemans2015} but uses a different set $X \subset \Ptop$ of distinguished points. Instead of the set of vertices of integral parallelepipedes, our theorem uses the set of vertices $V_I$ of the integer polytope.
\begin{theorem} \label{thm-structure}
Let $P = \{ x \in \mathbb{R}^{d} \mid Ax \leq c \}$ be a polytope with $A \in \ZZ^{m \times d}, c \in \ZZ^{d}_{\geq 0}$ and let $V_I \subseteq \Ptop \cap \ZZ^{d}$ be the set of vertices of the integer polytope $P_I$ with $Conv(V_I) \cap \ZZ^d = \Ptop \cap \ZZ^d$. Then for any vector $\pp \in int.cone(\Ptop \cap \ZZ^d)$, there exists an integral vector $\lambda \in \ZZ^{\Ptop \cap \ZZ^d}_{\geq 0}$ such that $\pp = \sum_{p \in \Ptop \cap \ZZ^d} \lambda_p p$ and 
\begin{enumerate}
\item $\lambda_p \leq 2^{2^{O(d)}} \quad \forall p \in (\Ptop \cap \ZZ^d) \setminus V_I$
\item $|supp(\lambda) \cap V_I| \leq d \cdot 2^d$
\item $|supp(\lambda) \setminus V_I| \leq 2^{2d}$
\end{enumerate}
\end{theorem}
This theorem finally shows that for arbitrary polytopes $\Ptop$ and any $b \in int.cone(\Ptop \cap \ZZ^d)$, the vertex distance $Dist(b)$ is bounded by $2^{2^{O(d)}}$ and hence independent of the number of inequalities $m$ and the largest entry $\Delta$ in the description of $\Ptop$.

Recall that a parameterized problem with parameter $p$ and input $I$ is called fixed parameter tractable (fpt) if there exists an algorithm with running time $O(f(p) \cdot enc(I)^{O(1)})$ for some computable function $f$ of $p$ which is independent of $I$ and $enc(I)$ is the encoding length of instance $I$. We refer to the book of Downey and Fellows \cite{Downey99} for more details on parameterized complexity.
As a consequence of our structure theorem, we present in Section 2 an algorithm for the bin packing problem with a running time of $|V_I|^{2^{O(d)}} \cdot \log(\Delta)^{O(1)}$, where $\Delta$ is the maximum over all multiplicities $b$ and denominators in $s$. Since $|V_I| \geq d+1$ this is an fpt-algorithm parameterized by the number of vertices of the integer knapsack polytope $V_I$.
\begin{theorem} \label{thm-fpt}
    The bin packing problem can be solved in fpt-time parameterized by the number of vertices $V_I$ of the integer knapsack polytope.
\end{theorem}
This theorem shows that the bin packing problem can be solved efficiently when the underlying knapsack polytope has an easy structure i.e. has not too many vertices. However, since the total number of vertices is bounded by $O(\log \Delta)^d$ \cite{hayes} the algorithm has a worst case running time of $(\log \Delta)^{2^{O(d)}}$, which is identical to the running time of the algorithm by Goemans and Rothvo\ss~\cite{goemans2015}.

\subsection{Related results}
The bin packing problem is one of the most fundamental combinatorial problems in computer science. It has been very well studied in the literature, mostly in the context of approximation. A major contribution was given by Karmarkar and Karp \cite{karmarkar1982}. They presented a polynomial time approximation algorithm with a guarantee of $OPT + O(\log^2 (OPT))$. Very recently, this famous result by Karmarkar and Karp was improved by Rothvo\ss~\cite{Rothvoss2013} who presented an algorithm with guarantee $OPT + O(\log OPT \log \log (OPT))$ and later by Hoberg and Rothvo\ss~\cite{rothvoss2015} who improved the guarantee further to $OPT + O(\log (OPT))$. Concerning the bin packing problem when the number of different item sizes $d$ is constant, Jansen and Solis-Oba~\cite{jansen2011} presented an approximation algorithm with a guarantee of $OPT+1$. Their algorithm has a running time of $2^{2^{O(d)}} \cdot enc(I)^{O(1)}$ and therefore is fpt in the number of different item sizes $d$. Finally, as mentioned above, Goemans and Rothvo\ss~\cite{goemans2015} presented their polynomial time algorithm for the bin packing problem with running time $(\log \Delta)^{2^{O(d)}}$. 

In a very recent work, Onn \cite{Onn15} discussed the problem of finding a vector $\lambda \in \ZZgeq^{\Ptop \cap \ZZ^d}$ with $b = \sum_{p \in \Ptop \cap \ZZ^d} \lambda_p p$ for given $b \in int.cone(\Ptop \cap \ZZ^d)$. He presented an algorithm for the case that the polytope $\Ptop = \{ x \in \QQ^d \mid Ax \leq c\}$ has a specific shape. In the case that the matrix $A$ is totally unimodular he gave a polynomial time algorithm even in the case that the dimension $d$ is variable.

\section{Proof of the main theorem}
Given simplex $S = Conv(B_0, B_1, \ldots , B_d)$ for $B_i \in \ZZ^{d}_{\geq 0}$ and a vector $\pp \in int.cone(S \cap \ZZ^d)$. We consider integral points in $cone(B)$ generated by the vertices $B = \{B_0, B_1, \ldots , B_d\} = V_I$ of the simplex $S$. For convenience, we denote by $B$ also the matrix with columns $B_0, B_1 , \ldots , B_d$. As our main subject of investigation, we consider the parallelepiped
\begin{align*}
    \PP = \{ x_0 B_0 + x_1 B_1 + \ldots + x_d B_d \mid x_i \in [0,1] \}.
\end{align*}
By definition of $S = \{ x_0 B_0 + x_1 B_1 + \ldots + x_d B_d \mid x_i \in [0,1], \sum_i x_i = 1 \}$ we have that $S \subset \Pi$.
Furthermore, one can easily see that $cone(B)$ can be partitioned into parallelepipedes $\Pi$ (see figure \ref{fig-cone} with $B_0 = 0$), as each point $\pp = x_0 B_0 + x_1 B_1 + \ldots + x_d B_d \in cone(B) \cap \ZZ^d$ for some $x \in \QQgeq^d$ can be written as the sum of an integral part $B x^{int} = \lfloor x_0 \rfloor B_0 + \ldots + \lfloor x_d \rfloor B_d$ and a fractional part $[ B x ] = \{ x_0 \} B_0 + \ldots + \{ x_d \} B_d \in \PP$ (we denote the fractional part of some $v \in \QQ$ by $\{ v \} = v - \lfloor v \rfloor$ and for some vector $x \in \QQ^d$ we denote by $\{ x \}$ the vector $(\{ x_0 \}, \ldots , \{ x_d\})^T $.
\begin{figure}
    \centering
\begin{tikzpicture}[scale=0.88]
        \tikzset{
        >=stealth',
        help lines/.style={dashed, thick},
        axis/.style={<->},
        important line/.style={thick},
        connection/.style={thick, dotted},
        }
    \coordinate (y) at (0,5.2);
    \coordinate (x) at (7.2,0);
    \coordinate (b1) at (0.5,1.5);
    \coordinate (b14) at (2,6);
    \coordinate (b12) at (1,3);
    \coordinate (b2) at (2.5,.5);
    \coordinate (b24) at (10,2);
    \coordinate (b1b2) at (3,2);
    \coordinate (g1) at (1.2,0.85);
    \coordinate (g2) at (2.4,1.7);
    \coordinate (g3) at (3.6,2.55);
    \coordinate (b22) at (5,1);
    \coordinate (b13) at (1.5,4.5);
    \coordinate (b23) at (7.5,1.5);
    \coordinate (b12b2) at (3.5,3.5);
    \coordinate (b1b22) at (5.5,2.5);

    \draw[axis] (y) -- (0,0) --  (x);

    \draw[dotted] (0,0) -- (g1);
    \draw[dotted] (g2) -- (g1);
    \draw[dotted] (g2) -- (g3);

    \draw[] (0,0) -- (b1);
    \draw[] (0,0) -- (b2);
    \draw[] (b1) -- (b1b2);
    \draw[] (b2) -- (b1b2);
    \draw[] (b1b22) -- (b1b2);
    \draw[] (b2) -- (b22);
    \draw[] (b1) -- (b12);
    \draw[] (b22) -- (b23);
    \draw[] (b22) -- (b1b22);
    \draw[] (b12b2) -- (b1b2);
    \draw[] (b12) -- (b12b2);
    \draw[] (b12) -- (b13);

    \fill(b1) circle (1.5pt);
    \fill(b2) circle (1.5pt);
    \fill(g1) circle (1.5pt);
    \fill(g2) circle (1.5pt);
    \fill(g3) circle (1.5pt);

    \fill[black,font=\footnotesize]
                    (b1) node [left] {$B_1$}
                    (b2) node [right] {$B_2$}
                    (0,0) node [left,below] {$B_0 = 0$}

                    (g1) node [left] {$\gamma$}
                    (g2) node [left] {$2 \gamma$}
                    (g3) node [left] {$3 \gamma$}
                    (2,1) node [left] {\bf \large  $\Pi$};

    \fill[fill opacity=0.2] (b1) -- (b2) -- (0,0) -- cycle;
    \fill[fill opacity=0.2] (b1) -- (b12) -- (b1b2) -- cycle;
    \fill[fill opacity=0.2] (b2) -- (b1b2) -- (b22) -- cycle;
    \fill[fill opacity=0.2] (b22) -- (b1b22) -- (b23) -- cycle;
    \fill[fill opacity=0.2] (b1b22) -- (b1b2) -- (b12b2) -- cycle;
    \fill[fill opacity=0.2] (b12) -- (b12b2) -- (b13) -- cycle;

    \draw[thick, dotted] (b13) -- (b14);
    \draw[thick, dotted] (b23) -- (b24);
    
  \end{tikzpicture}

      \caption{Partitioning $Cone(B)$}
    \label{fig-cone}
\end{figure}
For vector $\pp \in Cone(B)$ with $b = Bx$ let $[ \pp ] = [ B x ]$. We say that two points $\pp, \pp' \in Cone(B)$ are \emph{equivalent} if $[ \pp ] = [ \pp' ]$.

For the proof of the main theorem, we consider a $\lambda \in \ZZgeq^{\Ptop \cap \ZZ^d}$ with $b = \sum_{s \in S} \lambda_s s$ and suppose that $\lambda$ does not fulfill property (1) of Theorem \ref{thm-main}. Then there exists a $\gamma \in (\Ptop \cap \ZZ^d) \setminus V_I$ with big weight i.e. $\lambda_\gamma \geq 2^{2^{\Omega(d)}}$.
The key idea of the proof is that we consider the set of multiplicities $\gamma, 2 \gamma , 3 \gamma , \ldots $ of the vector $\gamma$. Our goal is to find a possibly small multiplicity $K>1$ such that $K \gamma$ is equivalent to a point $\delta$ in the convex hull $S$. Hence, weight on $\gamma$ can be shifted to the vertices $B_0, \ldots , B_d$ of $S$. Then $K \gamma$ can be written as the sum of vertices $\sum \Lambda_i  B_i$ plus some $\delta \in S \cap \ZZ^d$ (see Lemma \ref{lem-konfig} for a detailed proof). In figure \ref{fig-cone} we have that $3 \gamma$ is equivalent to a point in the simplex (as remarked by the grey areas) and hence in that case $3 \gamma = \delta B_1 + B_2$ for some $\delta \in S$. Before we are ready to prove the existence of a small multiplicity $K$, we give some definitions and observations.



Instead of multiplicities of $\gamma \in S \cap \ZZ^d$, we consider multiplicities of a vector $x \in [0,1)^{d+1}$ in the unit cube with $\gamma = x_0 B_0 + x_1 B_1 + \ldots + x_d B_d = Bx$.
\begin{definition}
Consider multiplicities $x, 2x, 3x, \ldots$ of a vector $x \in [0,1)^{d+1}$ with $\sum x_i = 1$. We say components $i$ \emph{jumps} at $K$ if $\lceil K x_i \rceil > \lceil (K-1) x_i \rceil$.
We define 
\begin{align*}
Level(Kx) =  \sum_{i=0}^d \{ K x_i \}.
\end{align*}
\end{definition}
Note that $Level(Kx)$ is always integral as $Level(Kx) = \sum_{i=0}^d K x_i - \sum_{i=0}^d \lfloor K x_i \rfloor$ and both terms $\sum_{i=0}^d K x_i$ and $\sum_{i=0}^d \lfloor K x_i \rfloor$ are integral.
The following lemma shows that we obtain the desired decomposition of $K \gamma$ if $Level(K x) = 1$.
\begin{lemma} \label{lem-konfig}
Let $\gamma \in S \cap \ZZ^d$ be the vector with $\gamma = Bx$ for $x \in [0,1)^{d+1}$. If $Level(Kx) = 1$, then there exists a $\Lambda \in \ZZ^{d+1}_{\geq 0}$ and a $\delta \in S \cap \ZZ^d$ such that
\begin{align*}
K \gamma = \delta + \sum_{i=0}^d \Lambda_i B_i.
\end{align*}
\end{lemma}
\begin{proof}
As above, we split every component $i$ of $Kx \in \QQ^{d+1}_{\geq 0}$ into an integral part $Kx^{int}_i = \lfloor K x_i \rfloor$ and a fractional part $\{ K x_i \}$. Then $Kx = Kx^{int} + \{ Kx \}$ and we set $\delta = B (\{ K x \}) = [ B(Kx) ]$.

\begin{observation} $\delta \in \ZZ^d$. \end{observation}
Since $K \gamma  = K B x \in \ZZ^d$ and $K B x^{int} \in \ZZ^d$ we obtain that $\delta = B (\{ K x \}) = B (Kx) - B (K x^{int})$ is integral and therefore $\delta \in \ZZ^d$.

\begin{observation} $\delta \in S$. \end{observation}
Since $Level(Kx) = 1$ we obtain that $\sum_{i=0}^d \{ K x_i \} = 1$ and $\delta = B ( \{ K x \}) = ( \{ K x_0 \}) B_0 + \ldots + (\{ K x_d \}) B_d$, we can state $\delta$ as a convex combination of $B_0, B_1 , \ldots B_d$. Therefore $\delta \in S$.

Finally, we can decompose $K \gamma$ into
\begin{align*}
    K \gamma = K B x = B (\{ K x_d \}) + B (K x^{int}) = B (K x^{int}) + \delta =  \delta + \sum_{i=0}^d \Lambda_i B_i
\end{align*}
for some $\Lambda \in \ZZ^{B}_{\geq 0}$.
\end{proof}

The following lemma gives a correlation between the level of some point $Kx$ and the number of jumps.
\begin{lemma} \label{lem-level}
Let $J$ be the number of jumps at $K$, then
\begin{align*}
Level(Kx) = Level((K-1)x)+1 - J.
\end{align*}
\end{lemma}
\begin{proof}
    For every component $i$ of $Kx$ which jumps, we obtain that $\{ Kx_i \} = \{ (K-1)x_i + x_i \}= \{ (K-1)x_i \} +x_i -1$.
    Hence 
    \begin{align*}
        & Level(Kx) = \sum_{i=0}^d \{ K x_i \} \\ &=  \sum_{i\text{ jumps at }K} ( \{ (K-1)x_i \} +x_i -1) + \sum_{i\text{ does not jump at }K} ( \{ (K-1)x_i \} + x_i) \\&=  \sum_{i=0}^d \{ (K-1) x_i \} + \sum_{i=0}^d x_i  -J = Level((K-1)x) +1 -J.
    \end{align*}
\end{proof}

\begin{theorem} \label{thm-main1}
Given $x \in [0,1)^{d+1}$ with $Level(x) = \sum_{i=0}^d x_i = 1$. Then there is a $K \in \mathbb{N}_{>1}$ with $K \leq 2^{2^{O(d)}}$ such that $Level(Kx) = \sum_{i=0}^d \{ K x_i \}  = 1$.
\end{theorem}
\begin{proof}
We suppose that $\frac{1}{2} > x_0 \geq x_1 \geq \ldots \geq x_d$. In the case that $x_0 \geq \frac{1}{2}$ we obtain that $Level(2x) \leq 1$ and are done. Let $\bar{d} \leq d$ be the smallest index such that 
\begin{align*}
\sum_{i=\bar{d}+1}^{d}x_{i} < \frac{1}{\XX(\bar{d})} x_{\bar{d}},
\end{align*}
where $\XX(\bar{d}) = \prod_{i=0}^{\bar{d}} p_i$ and $p_i = \lceil \frac{1}{x_i} \rceil$. Intuitively, $\bar{d}$ is chosen such that there is a major jump from $x_{\bar{d}}$ to $x_{\bar{d}+1}$, i.e. $x_{\bar{d}} >> x_{\bar{d}+1}$. Note that the above equation is always fulfilled for $\bar{d} = d$ and since $\frac{1}{\XX(\bar{d})} x_{\bar{d}} < x_{\bar{d}} < 1/2$ we have that $\bar{d} \geq 1$.
First, we prove the following lemma to give bounds for $\XX(\bar{d})$ and $x_{\bar{d}}$
\begin{lemma}
Assuming for every $0 \leq j \leq \bar{d}$ that $\sum_{i=j+1}^{d}x_{i} \geq \frac{1}{\XX(j)} x_{j}$, then the following parameters can be bounded by
\begin{itemize}
\item $\XX(\bar{d}) \leq2^{2^{O(d)}}$ and
\item component $x_{\bar{d}} \geq \frac{1}{2^{2^{\Omega(d)}}}$.
\end{itemize}
\end{lemma}
\begin{proof}
For $j=0$ we know that $x_0 \geq \frac{1}{d+1}$ as $x_0$ is the largest component. This implies also that $\XX(0) \leq d+1$. We suppose by induction that for every $j \leq \bar{d}$,
\begin{align*}
x_{j} \geq (2^{j2^{2j}} \cdot d^{2^{2j}})^{-1} 
\end{align*}
and 
\begin{align*}
\XX(j) \leq (2^{j2^{2j+1}} \cdot d^{2^{2j+1}}).
\end{align*}
Since $j \leq \bar{d}$ we obtain that $\sum_{i=j+1}^{d} x_{i} \geq \frac{1}{\XX(j)} x_j$ and since the coefficients are sorted in non-increasing order we get $d x_{j+1} \geq \frac{1}{\XX(j)} x_j$. Using the induction hypothesis this gives 
\begin{align*}
x_{j+1} & \geq (2^{j2^{2j+1}} d^{2^{2j+1}})^{-1} \cdot (2^{j2^{2j}} d^{2^{2j}})^{-1} \cdot d^{-1} \\
&\geq (2^{j2^{2j+1} + j2^{2j}} \cdot d^{2^{2j} + 2^{2j+1} +1})^{-1}\\ 
& > (2^{2 \cdot j2^{2j+1}} \cdot d^{2 \cdot 2^{2j+1}})^{-1}\\
&= (2^{j2^{2(j+1)}} \cdot d^{2^{2(j+1)}})^{-1}
\end{align*}
Product $\XX(j+1)$ can be bounded as follows: 
\begin{align*}
\XX(j+1) & = \lceil \frac{1}{x_{j+1}} \rceil \XX(j) \leq (\frac{1}{x_{j+1}}+1) \XX(j)\\
& \leq (2^{(j+1)2^{2(j+1)}} d^{2^{2(j+1)}} +1) \cdot 2^{j2^{2j+1}} d^{2^{2j+1}}\\
& < 2^{(j+1)2^{2(j+1)}+1} d^{2^{2(j+1)}} \cdot 2^{j2^{2j+1}} d^{2^{2j+1}}\\
& = 2^{(j+1)2^{2(j+1)} + 1 + j2^{2j+1}} \cdot d^{2^{2(j+1)} + 2^{2j+1}}\\
& < 2^{2 \cdot (j+1)2^{2(j+1)}} \cdot d^{2 \cdot 2^{2(j+1)}}\\
& = 2^{\cdot (j+1)2^{2(j+1) +1}} \cdot d^{2^{2(j+1)+1}}
\end{align*}
As a result we obtain that $\XX(\bar{d}) \leq (2^{\bar{d}2^{2\bar{d}+1}} \cdot d^{2^{2\bar{d}+1}}) = 2^{2^{O(d)}}$ and $x_{\bar{d}} \geq (2^{\bar{d}2^{2\bar{d}}} \cdot d^{2^{2\bar{d}}})^{-1} = \frac{1}{2^{2^{\Omega(d)}}}$.
\end{proof}

Let $\XX = \XX(\bar{d})$, for each component $0 \leq i \leq d$, we define the \emph{distance} $D_i(K)$ of a multiplicity $K$  by $D_i(K) = j$, where $j \geq 0$ is the smallest integer such that component $i$ jumps at $K+j$. Note that $D_i(K)$ is bounded by $p_i$ as $p_i x_i \geq 1$. We say $Kx \equiv K'x$ if for every $0 \leq i \leq \bar{d}$ the distance $D_i(K) = D_i(K')$.
Consider elements $x, 2x, \ldots, (\XX+1) x$. Since the number of equivalence classes is bounded by $\XX =\prod_{i=1}^{\bar{d}} p_i$, there exist two elements $Kx, (K+Z)x$ with $K,Z \in \ZZ_{\geq 1}$ and $K, (K+Z) \leq \XX+1$ such that $Kx \equiv (K+Z)x$. We will see that the equivalence of two multiplicities implies the existence of a multiplicity $M>1$ with $Level(Mx)=1$.

First, we argue about the level of multiplicity $Z-1$. Note that since $\sum_{i=0}^d x_i = 1$ we have that $Level(Kx) = \sum_{i=0}^d \{ K x_i \}$ for every multiplicity $K \in \ZZ_{\geq 1}$.
The case that $Level((Z-1)x) \geq \bar{d}+2$ is not possible since
\begin{align*}
 Level((Z-1)x) = \sum_{i=0}^d \{ (Z-1) x_i \} \leq \bar{d}+1 + (Z-1) \sum_{i=\bar{d}+1}^{d} x_{i} \leq \bar{d} +1 + (\XX-1) \sum_{i=\bar{d}+1}^{d} x_{i} \\ \leq \bar{d} +1 + \frac{\XX-1}{\XX} x_{\bar{d}} < \bar{d} + 2.
\end{align*}

\begin{case}[1] Suppose $Level((Z-1)x) = \bar{d}+1$ (for $Z \geq 2$). \end{case}

\begin{case}[1a] Suppose $\{ (Z-1)x_i \} \geq 1-x_i$ for all $i = 0,1, \ldots, \bar{d}$. \end{case} In this case every component $0 \leq i \leq \bar{d}$ jumps at $Z$. By Lemma \ref{lem-level}, we can bound the level of $Zx$ by $Level(Zx) \leq Level((Z-1)x)+1 - (\bar{d}+1) = 1$.

\begin{case}[1b] There is an $0 \leq i \leq \bar{d}$ such that $\{ (Z-1) x_i \} < 1-x_i$. \end{case} In this case $\sum_{i=0}^{\bar{d}} \{ (Z-1) x_i \} \leq \bar{d}+1 - x_{i} \leq \bar{d}+1 - x_{\bar{d}}$ and we obtain
\begin{align*}
Level((Z-1)x) =  \sum_{i=0}^d \{ (Z-1) x_i \} < \bar{d} +1 - x_{\bar{d}} + \sum_{i=\bar{d}+1}^{d} (Z-1) x_{i} \leq \bar{d} +1 - x_{\bar{d}} + \frac{1}{X} x_{\bar{d}} < \bar{d} +1
\end{align*}
which contradicts the assumption of Case 1.
\begin{case}[2] Suppose $Level((Z-1)x) \leq \bar{d}$. \end{case}
Since $K \equiv K+Z$ we know that every component $i$ jumps at $K+D_i$ and $K+Z+D_i$, hence for every $0 \leq i \leq \bar{d}$ we obtain $\{ (K+D_i) x_i \} < x_i$ and $\{ (K+Z+ D_i)x_i \} < x_i$. Let $\{ (K+D_i) x_i \} = \alpha_1$ and $\{ (K+Z+D_i) x_i \} = \alpha_2$ for some $\alpha_1, \alpha_2 < x_i$. Then $\{ Zx_i \} = \{ \alpha_2 -\alpha_1 \}$ and hence $\{ Zx_i \} = \alpha_2 -\alpha_1$ if $\alpha_1 \leq \alpha_2$ and $\{ Zx_i \} = 1 + \alpha_2 -\alpha_1$ if $\alpha_1 > \alpha_2$. Since $\alpha_1, \alpha_2 < x_i$ we have $\{ Z x_i \} < x_i$ or $\{ Z x_i \} \geq 1 - x_i$. Hence every component $0 \leq i \leq \bar{d}$ jumps at $Z x_i$ or at $(Z+1)x_i$.
We obtain by Lemma \ref{lem-level} that $Level((Z+1)x) \leq Level((Z-1)x) + 2 - (\bar{d} +1) \leq \bar{d}+ 2 - (\bar{d} +1) = 1$.
\end{proof}

\subsubsection*{Proof of the main Theorem \ref{thm-main}}
Consider the vector $\pp \in int.cone(S \cap \ZZ^d)$. Let $\lambda \in \ZZ^{S \cap \ZZ^d}$ be the integral vector with $\pp = \sum_{s \in S \cap \ZZ^d} \lambda_s s$.

Assume there is a component $\gamma \in (S \cap \ZZ^d ) \setminus B$ with high multiplicity i.e. $\lambda_{\gamma} = 2^{2^{\Omega(d)}}$. Since $\gamma \in S$, there is a $x \in \QQgeq^{d+1}$ with $\sum_{i=0}^d x_i = 1$ and $\gamma = x_0 B_0 + x_1 B_1 + \ldots x_d B_d$. By Theorem \ref{thm-main1} there exists a multiplicity $K = 2^{2^{O(d)}} > 1$ such that $Level(K x) = 1$.
According to Lemma \ref{lem-konfig}, there exists a $\delta \in S \cap \ZZ^{d}_{\geq 0}$ such that $K \gamma = \delta + \sum_{i=0}^d \Lambda_i B_i$ for some $\Lambda \in \ZZgeq^{d+1}$. Then we can construct a $\lambda' \in \ZZ^{d+1}_{\geq 0}$ with $\pp = \sum_{s \in S \cap \ZZ^d} \lambda'_p p$, which has more weight in $B$ as $K>1$: 
\begin{align*}
\lambda' = \begin{cases}\lambda_\gamma - K \\
    \lambda_\delta +1 \\
    \lambda_{B_i} + \Lambda_i & \forall B_i \in B \\
    \lambda_s & \forall s \in (S\cap \ZZ^d) \setminus (B \cup \{ \gamma, \delta \})
\end{cases}
\end{align*}
Since $K >1$, the resulting vector $\lambda'$ has a decreased vertex distance as the sum $\sum_{p \in (\Ptop \cap \ZZ^d) \setminus V_I} \lambda_p$ is reduced at least by $1$.
Using Theorem \ref{thm-supp} applied to components $i \in (S \cap \ZZ^d) \setminus B$, we can construct a $\lambda''$ with $|supp( \lambda'' \setminus B)| \leq 2^d$.
In the case that there is another component $\gamma \in (S \cap \ZZ^d ) \setminus B$ with multiplicity $\lambda''_{\gamma} = 2^{2^{\Omega(d)}}$ we can iterate this process. In the other case, solution $\lambda''$ fulfills the proposed properties.

\subsubsection*{Proof of the structure Theorem \ref{thm-structure}}
\begin{proof} Given polytope $\Ptop = \{ x \in \QQ^d \mid Ax \leq c\}$ for some matrix $A \in \ZZ^{m \times d}$ and a vector $c \in \ZZ^d$ and let $\Ptop_I$ be the integer polytope with vertices $V_I$. The structure theorem follows easily by decomposing the polytope $\Ptop$ into simplices $S$ of the form $S = Conv(B_0, B_1, \ldots B_d)$ for $B_i \in V_i$.

By Caratheodory's Theorem, there exist for each $\gamma \in \Ptop$ vertices $B_0, B_1, \ldots , B_{d} \in V_I$ and a $x \in \QQgeq^{d+1}$ with $\sum_{i=0}^{d} x_i = 1$ such that $\gamma = x_0 B_0 + \ldots + x_{d} B_{d}$.
Consider the vector $\pp \in int.cone(\Ptop \cap \ZZ^d)$. Let $\lambda \in \ZZ^{\Ptop \cap \ZZ^d}$ be the integral vector with $\pp = \sum_{p \in \Ptop \cap \ZZ^d} \lambda_p p$. By Theorem \ref{thm-supp}, we can assume that $supp(\lambda) \leq 2^d$ and hence there are at most $2^d$ simplices $S^{(k)} = Conv(B^{1}_0, B_{1}^{(k)}, \ldots B_{d}^{(k)})$ for $k = 1, \ldots , 2^d$ which contain a point $\gamma \in \Ptop$ with $\lambda_\gamma > 0$.
Finally, we can apply our main Theorem \ref{thm-main} to every simplex $S^{(k)}$ for $k = 1, \ldots , 2^d$ to obtain a vector $\lambda' \in \ZZgeq^{\Ptop \cap \ZZ^d}$ which fulfills the above properties.
\end{proof}


\subsection{Algorithmic application}

\subsubsection*{Computing $V_I$ in fpt-time}
Given Polytope $\Ptop = \{ x \in \QQ^d \mid Ax \leq c\}$ for some matrix $A \in \ZZ^{m \times d}$ and a vector $c \in \ZZ^d$ such that all coefficients of $A$ and $c$ are bounded by $\Delta$. Cook et al. \cite{cook1992} proved that the number of vertices $V_I$ of the integer polytope $\Ptop_I$ is bounded by $m^d \cdot O((\log \Delta))^d$. In the following we give a brief description on how the set of vertices $V_I$ can be computed in time $|V_I| \cdot d^{O(d)} \cdot (m \log(\Delta))^{O(1)}$ and therefore in fpt-time parameterized by the number of vertices $|V_I|$. For a detailed description of the algorithm we refer to the thesis of Hartmann \cite{hartmann}.

Given at timestep $t$ a set of vertices $V_t \subset V_I$ and the set of facets $F^{(t)} = \{F_1, \ldots , F_\ell \}$ of $conv(V_t)$ corresponding to half-spaces $H_i = \{ x \mid n_i x \leq c_i\}$ for normal vectors $n_1, \ldots , n_\ell \in \QQ^d$ and constants $c_1, \ldots , c_d \in \ZZ$ with $conv(V) = \cap_{i=1}^\ell H_i$. Consider for every $1 \leq i \leq \ell$ the polytope $\Ptop \cap H_{i}^-$ for a halfspace  $H_{i}^- = \{x \mid x \in \Ptop, n_i x \geq c_i \}$.  Compute with Lenstra's algorithm \cite{Lenstra1983} a solution $x^* \in \Ptop \cap \ZZ^d$ of the ILP $\max \{ n_i x \mid x \in (\Ptop \cap \ZZ^d), n_i x \geq c_i \}$. In the case that $n_i x^* > c_i$, solution $x^*$ does not belong to $Conv(V_t)$. Assuming that $x^*$ is a vertex of the integer polytope of $\Ptop \cap H_{i}^-$, solution $x^*$ is also a vertex of the integer polytope $\Ptop_I$. We can add $x^{*}$ to the set of existing vertices $V_t \subset V_I$, construct the increased set of facets $F^{(t+1)}$ of $Conv(V_t \cup \{ x^* \})$ and iterate the procedure. In the case that there is no solution $x^*$ with $n_i x^* > c_i$ for any $1 \leq i \leq \ell$, we have that $\Ptop_I = Conv(V_t)$ and are done.

\subsubsection*{Bin Packing in fpt-time}
In the following we describe the algorithmic use of the presented structure theorem \ref{thm-structure}. Therefore, we follow the approach by Goemans and Rothvo\ss~\cite{goemans2015}.
\begin{theorem}
Given polytopes $\Ptop,\mathcal{Q} \subset \QQ^d$, one can find a $y \in int.cone(\Ptop \cap \ZZ^d) \cap \mathcal{Q}$ and a vector $\lambda \in \ZZgeq^{\Ptop \cap \ZZ^d}$ such that $\pp = \sum_{p \Ptop \cap \ZZ^d} \lambda_p p$ in time $|V_I|^{2^{O(d)}} enc(\Ptop)^{O(1)} enc(\mathcal{Q})^{O(1)}$, where $enc(\Ptop), enc(\mathcal{Q})$ is the encoding length of the polytope $\Ptop, \mathcal{Q}$ or decide that no such $y$ exists.
\end{theorem}
\begin{proof}
Let $\Ptop = \{ x \in \QQ^d \mid Ax \leq c\}$ and $\mathcal{Q} = \{ x \in \QQ^d \mid \tilde{A}x \leq \tilde{c} \}$ be the given polytopes for a matrix $A \in \ZZ^{m \times d}$ and a matrix $\tilde{A} \in \ZZ^{\tilde{m}\times \tilde{d}}$.
First, we compute the set of vertices of $\Ptop_I$ in time $|V_I|^{2^{O(d)}} enc(\Ptop)^{O(1)} enc(\mathcal{Q})^{O(1)}$ as described above. Suppose that there is a vector $\pp \in int.cone(\Ptop \cap \ZZ^d) \cap \mathcal{Q}$, then by Theorem \ref{thm-structure}, we know there is a vector $\lambda \in \ZZgeq^{\Ptop \cap \ZZ^d}$ with $\pp = \sum_{p \in \Ptop \cap \ZZ^d} \lambda_s s$ such that 
\begin{enumerate}
\item $\lambda_p \leq 2^{2^{O(d)}} \forall p \in (\Ptop \cap \ZZ^d) \setminus V_I$,
\item $|supp(\lambda) \cap V_I| \leq d \cdot 2^d$,
\item $|supp(\lambda) \setminus V_I| \leq 2^{2d}$.
\end{enumerate}
At the expense of a factor $\binom{|V_I|}{d2^{d}} = |V_I|^{2^{O(d)}}$ we can guess the support $V_\lambda \subseteq V_I$ of $\lambda$ restricted to components $\lambda_p$ with $p \in V_I$ i.e. $V_\lambda = supp(\lambda) \cap V_I$. For each $p \in V_\lambda$ we use variables $\bar{\lambda}_p \in \ZZgeq$ to determine the multiplicities of $\lambda_p$. 
Furthermore, we guess the number of different points $p \not\in V_i$ used in $\lambda$ i.e. $k = | supp(\lambda) \setminus V_I| \leq 2^{2d}$. We use variables $x^{(j)}_i$ for $j = 1, \ldots , k$ to determine the points $p \not\in V_I$ and their multiplicity $\lambda_p$. Note that in the following ILP, we encode the multiplicity of a $\lambda_p$ with $p \not\in V_I$ binary, therefore the number of variables $x^{(j)}_i$  can be bounded by $2^{2d} \cdot \log (2^{2^{O(d)}}) = 2^{O(d)}$. And finally we use a vector $y \in \ZZ^d$ to denote the target vector in polytope $\mathcal{Q}$.
\begin{align*}
Ax^{(j)}_i  &\leq b \qquad \forall i = 1 , \ldots , 2^{O(d)} \text{ and } \forall j = 1 , \ldots , k\\
\sum_{p \in  V_\lambda} \bar{\lambda}_p p  + \sum_{j=1}^k \sum_{i=1}^{2^{O(d)}} 2^j x^{(j)}_i &= y\\
\tilde{A} y &\leq \tilde{b}\\
x_i &\in \ZZ^d \qquad i = 1, \ldots , k\\
\bar{\lambda}_p & \in \ZZgeq \qquad \forall p \in V_\lambda\\
y \in \ZZ^d
\end{align*}
Using the algorithm of Lenstra or Kannan (\cite{Lenstra1983},\cite{Kannan1987}) to solve the above ILP which has $k2^{O(d)} + d + d2^d = 2^{O(d)}$ variables and $mk+ d + \tilde{m} + d|V_\lambda| = m 2^{O(d)} + \tilde{m}$ constraints, takes time $(2^{O(d)})^{2^{O(d)}} \cdot (m 2^{O(d)} + \tilde{m})^{O(1)} \log(\bar{\Delta})^{O(1)} = 2^{2^{O(d)}} enc(\Ptop)^{O(1)} enc(\mathcal{Q})^{O(1)}$, where $\bar{\Delta} = \max\{ 2^{2^{O(d)}}, d! \Delta^d , \tilde{\Delta} \}$.
The total running time is hence of the form: $|V_I|^{2^{O(d)}} enc(\Ptop)^{O(1)} enc(\mathcal{Q})^{O(1)}$
\end{proof}
We can apply this theorem to the bin packing problem by choosing $\Ptop = \{ \begin{pmatrix}x \\ 1 \end{pmatrix} \in \QQgeq^{d+1} \mid s^T x \leq 1 \}$ and $\mathcal{Q} = \{ b \} \times [0,a]$ to decide if items $b$ can be packed into at most $a$ bins. Using binary search, on the number of used bins, we can solve the bin packing problem in time $|V_I|^{2^{O(d)}} \cdot \log(\Delta)^{O(1)}$, where $\Delta$ is the largest multiplicity of item sizes or the largest denominator appearing in an itemsize $s_1, \ldots , s_d$. Since $|V_I| \geq d+1$ this running time is fpt-time, parametrized by the number of vertices $|V_I|$ and therefore we obtain Theorem \ref{thm-fpt}.


\section{Lower Bound} \label{lower_bound}
In this section we give a construction of a bin packing instance $(s,b)$ with vertex distance $Dist(b) = 2^{2^{\Omega(d)}}$. We consider the case that all item sizes $s_1, \ldots , s_d$ are of the form $s_i = \frac{1}{a_i}$ for some $a_i \in \ZZ_{\geq 1}$. In this case, all vertices of the knapsack polytope $\Ptop = \{ x \in \QQgeq \mid s_1 x_1 + \ldots + s_d x_d \leq 1 \}$ are of the form $B_i = (0,\ldots,0, a_i ,0, \ldots , 0)^T$ and therefore integral. We obtain that $\Ptop_I = \Ptop = Conv(0,B_1, \ldots , B_d)$.

Our approach for the proof of the lower bound is as follows: We prove the existence of a parallelepipped $\Pi = \{ x_1 B_1 + \ldots + x_d B_d \mid x_i \in [0,1] \}$ with a special element $g \in (\Ptop \cap \ZZ^d) \setminus V_I$ such that the unique optimal packing of the bin packing instance $K \cdot g$ (for a possibly large multiplicity $K \in \ZZgeq$) is to use $K$ times the configuration $g$. In this case, weight can not be shifted to the vertices $B_1, \ldots , B_d$ and hence the instance $K g$ implies a vertex distance of $Dist(Kg) = K$. We show that the special element $g$ can be determined by a set of modulo congruences. Therefore, we are able to use basic number theory to construct a bin packing instance with double exponential vertex distance.

First, we take a close look at the parallelepipped $\Pi$ and $Cone(B)$. Recall that two points $\pp, \pp' \in Cone(B)$ are \emph{equivalent} if $[ \pp ] = [ \pp' ]$. Each point $\pp \in Cone(B)$ is equivalent to a point in the parallelepiped $\PP$.
\begin{lemma} \label{lem-group}
Using operation $+$ defined by $p+p' = [ p + p' ]$ for some $p,p' \in \PP \cap \ZZ^d$ then $G(\PP) = (\PP \cap \ZZ^d, +)$ is an abelian group with $|G(\PP)| = |det(B)|$ many elements.
\end{lemma}
The proof that $G(\Pi)$ is a group can be easily seen, since $G(\Pi)$ is the quotient group of $\ZZ^d$ and the lattice $\ZZ B_1 \oplus \ldots \oplus \ZZ B_d$.  For the fact that $|G(\Pi)| = det(B)$ we refer to \cite{Barvinok2007}.
In the considered case that $B_i = (0,\ldots,0, a_i ,0, \ldots , 0)^T$, the group $G(\Pi)$ is isomorphic to $(\ZZ /a_1 \ZZ) \times \ldots \times (\ZZ /a_d \ZZ)$. 
Recall that the set $\Ptop \cap \ZZ^d \subset \Pi \cap \ZZ^d$ represents all integral points of the knapsack polytope and each $\delta \in \Ptop \cap \ZZ^d$ therefore represents a way of packing a bin with items from sizes $s_1, \ldots , s_d$. We call $\delta \in \Ptop \cap \ZZ^d$ a \emph{configuration}.

In the following subsection \ref{sec-mirup}, we also give an easy observation on how the vertex distance is connected to the integrality gap of the bin packing problem.



\subsection{Preliminaries}
In this section we state some basic number theoretic theorems that we will use in the following. For details and proofs, we refer to the books of Stark \cite{stark1970} and Graham, Knuth and Patashnik \cite{Graham1994}.
\begin{theorem}[\cite{stark1970}]\label{coprime}
    Let $a_1, \ldots a_d \in \ZZ$ with $gcd(a_1 , \ldots , a_d) = 1$, then there exist $v_1, \ldots , v_d \in \ZZ$ such that
    \begin{align*}
        a_1 v_1 + \ldots a_d v_d = 1.
    \end{align*}
\end{theorem}
\begin{theorem}[Chinese remainder theorem \cite{stark1970}]\label{chinese}
    Suppose $a_1, \ldots , a_d \in \ZZ$ are pairwise coprime. Then, for any given sequence of integers $i_1, \ldots , i_d$, there exists an integer x solving the following system of simultaneous congruences.
        \begin{align*}
            x \equiv i_j \mod a_j \text{ for $1 \leq j \leq d$}.
        \end{align*}
    Furthermore, $x$ is unique $\mod \prod_{i=1}^d a_i$.
\end{theorem}
\begin{theorem}[\cite{stark1970}]\label{inverse}
    Given congruence $x \mod a$. If $x$ and $a$ are coprime, there exists an inverse element $x^{-1} \in \ZZ / a \ZZ$ such that $x x^{-1} \equiv 1 \mod a$.
\end{theorem}
Sylvester's sequence is defined by,
\begin{align*}
S_1 = 2\\
S_j = 1 + \prod_{i=1}^{j-1} S_i
\end{align*}
and has following properties (see \cite{Graham1994})
\begin{align*}
S_n \approx 1.264^{2^{n}}\\
\sum_{i=1}^{j-1} \frac{1}{S_i} = 1- \frac{1}{S_j -1}
\end{align*}

\subsection{Proof of the lower bound}
We start by defining the \emph{size} of an element $\pi \in \Pi$ by 
\begin{align*}
Size(\pi)= \sum_{i=1}^d s_i \pi_i.
\end{align*}
In our case, the sizes $s_i$ are given by $s_i = \frac{1}{a_i}$ and vectors $B_i = (0, \ldots ,0 , a_i , 0, \ldots ,0)^T$ for some $a_i  \in \ZZ_{\geq 1}$. Hence, the size of a $B_i$ equals to $1$ and for each $x \in [0,1)^d$ with $Bx = \pi$, we have that $\sum_{i=1}^d x_i = Size(\pi)$.
Since the matrix $B$ is a diagonal matrix with entries $a_i$, the determinant equals $det(B) = \prod_{i=1}^d a_i$.
We define for $1 \leq i \leq d$ that
\begin{align*}
    R_i = \frac{det(B)}{a_i} = \prod_{j \neq i} a_j.
\end{align*}
In the following lemma we show that the fractional value of the size $\{ Size(\Pi) \}$ is unique for every element $\pi \in G(\Pi )$.
\begin{lemma} \label{lem-size}
Given parallelepiped $\Pi = \{ x_1 B_1 + \ldots + x_d B_d \mid x_i \in [0,1)^d \}$ with $B_i = (0 , \ldots , 0 , a_i , 0 , \ldots ,0)^T$. If $a_1, \ldots , a_d$ are pairwise coprime, then for every $0 \leq a < det(B)$, there exists a unique vector $\pi \in \Pi \cap \ZZ^d$ and a vector $x \in [0,1)^d$ with $Bx = \pi$,  such that 
\begin{align*}
Size(\pi)= \sum_{i=1}^d x_i = z  + a / det(B),
\end{align*}
for some $z \in \ZZgeq$.
\end{lemma}
\begin{proof}
Since $a_1 \ldots , a_d$ are pairwise coprime, we have that $gcd(R_i, R_{i+1}) = \prod_{j \neq i,i+1} a_j$ and hence $gcd(R_1, \ldots , R_d) = 1$. By Theorem \ref{coprime}, there exist $v_1, \ldots , v_d \in \ZZ$ such that $v_1 R_1 + \ldots + v_d R_d = 1$. 
For $v'_i = v_i R_i \mod det(B)$ the sum $\sum_{i=1}^d v'_i \equiv 1 \mod det(B)$. Consider the vector $x = (\frac{v'_1}{det(B)}, \ldots , \frac{v'_d}{det(B)})^T$, then $\sum_{i=1}^d x_i = \sum_{i=1}^d \frac{v'_i}{det(B)} = \frac{1}{det(B)} + z$ for some $z \in \ZZgeq$.
The vector $Bx = \frac{v'_1}{det(B)} B_1 + \ldots + \frac{v'_d}{det(B)} B_d$ is integral since for every $1 \leq i \leq d$ the congruence $v'_i a_i \equiv v_i R_i a_i \equiv v_i det(B) \equiv 0 \mod det(B)$ holds and hence each item sizes $a_i x_i = \pi_i = \frac{v'_i}{det(B)} a_i \in \ZZgeq$.

Consider multiplicities $x, 2x, 3x, \ldots$. As above we can rewrite each element $Kx \in \QQgeq^d$ by $Kx = Kx^{int} + \{ Kx \}$ with $\{ Kx \} = (\{ K x_1 \}, \ldots , \{ K x_d \} )^T$ and $\{ K x_i \} < 1$, which implies that $B(\{ Kx \}) \in \Pi$. Since $B(Kx)$ is integral and $B(Kx^{int})$ is integral, the vector $B(\{ Kx \})$ is integral as well. Furthermore, the sum of all component $\{ Kx \}$ sums up to $\sum_{i=1}^d \{ K x_i \} = K \sum_{i=1}^d x_i - \sum_{i=1}^d \lfloor Kx_i \rfloor = \frac{K}{det(B)} + z$, for some $z \in \ZZgeq$. Hence for each multiplicity $Kx$ in $0x, x , 2x, \ldots , (det(B)-1)x$ there is a vector $B(\{ Kx \}) \in \Pi$ with $\sum_{i=1}^d \{ K x_i \} = z + \frac{K}{det(B)}$ and since $\Pi$ contains exactly $det(B)$ many elements (see Lemma \ref{lem-group}), each element of $G(\Pi)$ corresponds to a unique element of $0x ,x, 2x, \ldots , (det(B)-1) x$.
\end{proof}
Consider the specific element $g \in \Pi \cap \ZZ^d$ with fractional vector $x \in [0,1)^d$ such that $Bx = g$ and $Size(g) = \frac{det(B)-1}{det(B)} +z$. We call $g$ the \emph{full generator} of the group $G(\Pi)$.
\begin{corollary} \label{cor-generator}
    For every element $\pi \in G(\Pi)$ there exists a multiplicity $K$ such that $Kg = \pi$, i.e. the full generator $g$ generates the group $\Pi$ and hence $G(\PP) = <g>$ is a cyclic group. Element $Kg \in \Pi$ has a size of $z + \frac{det(B)-K}{det(B)}$ for some $z \in \ZZgeq$.
\end{corollary}
\begin{proof}
In the proof of the lemma above, we showed that the element $Bx$ with $\sum_{i=1}^d x_i = z + \frac{1}{det(B)}$ generates $G(\Pi)$ as each multiplicity $Kx$ of $x$ yields an element $B(\{ Kx \}) \in G(\Pi)$ of size $z + \frac{K}{det(B)}$. We consider the full generator $g$ with $g = B x'$ for some $x'$ with $\sum_{i=1}^d x'_i = z + \frac{det(B)-1}{det(B)}$ for some $z \in \ZZgeq$. By the same argument as before, the multiplicities $Kx'$ yield elements $\pi' = B(\{ Kx' \}) \in G(\Pi)$ with $Size(\pi') = z + \frac{det(B) - K}{det(B)}$ for some $z \in \ZZgeq$.
\end{proof}
As above, we consider multiplicities $K g$ of a vector $g \in \Pi$ and some $K >0$. We say that $K g \in cone(B)$ is \emph{unique} if $g \in \Ptop$ and $2 g, \ldots , K g \not\in \Ptop$ i.e. $g$ is a configuration and $2 g, \ldots , K g$ are not. In the following lemma we prove that if $K g$ is unique and $g$ is a full generator, then using $K$-times configuration $g$ is the unique optimal packing for instance $K g$.
\begin{lemma} \label{lem-unique}
Let $g$ be the full generator of $G(\Pi)$. If $K g \in cone(B)$ is unique, then there is no $\lambda \in \ZZgeq^{(\Ptop \cap \ZZ^d)}$ with $\lambda_{g} \neq K$ such that $\sum_{p \in \Ptop \cap \ZZ^d} \lambda_p p = K g$ and $|\lambda | = K$.
\end{lemma}
\begin{proof}
    Consider bin packing instance $Kg \in cone(B)$ and a packing of the instance into bins $1, \ldots , K$. Since $g$ contains items of size $\frac{det(B)-1}{det(B)}$, items in instance $Kg$ have a total size of $K \frac{det(B)-1}{det(B)}$ and therefore, the bins $1,\ldots, K$ have total free space of $\frac{K}{det(B)}$.
    Each bin configuration $c_1, \ldots , c_K$ of bins $1, \ldots , K$ belongs to $\Ptop$ and hence to $\Pi(G)$. By Cororllary \ref{cor-generator} for each $c_i$ there exists a multiplicity $K_i \in \ZZ_{\geq 1}$ such that $K_i g = c_i$.
    Assuming that $c_i \neq g$ and hence $K_i > 1$ we know that $K_i > K$ as by definition of the uniqueness of $K g$, elements $2g, \ldots , Kg$ are no configurations. However, a bin with configuration $K_i g = c_i \in \Ptop$ with $K_i >K$ has free space $\frac{K'}{det(B)}> \frac{K}{det(B)}$ and hence more free space than the total sum of free space in bins $1, \ldots , K$. Therefore, a configuration $\neq g$ can not appear in an optimal packing of the instance $K g$.
    The unique way of packing instance $Kg$ into $K$ bins is to use $K$ times configuration $g$.
\end{proof}
Consider the full generator $g = x_1 B_1 + \ldots + x_d B_d \in \Ptop$ with $x_i \geq 0$, we say $g$ has the \emph{long-run} property if $(1-\epsilon) \frac{1}{S_i} \leq x_i < \frac{1}{S_i}$ for $1\leq i \leq d-1$ and some $\epsilon < (\frac{1}{S_d-1})^2$, where $S_i$ is the $i$-th sylvester number. The following inequality gives a lower bound for $\nor{x}_1$:
\begin{align*}
    \sum_{i=1}^{d-1} x_i & \geq (1-\epsilon) \sum_{i=1}^{d-1} \frac{1}{S_i} = (1-\epsilon) (1-\frac{1}{S_d-1}) = 1- \epsilon - \frac{1-\epsilon}{S_d-1} \\ 
    & >  1- \frac{1}{(S_d-1)^2} - \frac{1}{S_d-1} =  1 - \frac{1}{(S_d-1)^2} + \frac{1}{(S_d-1)(S_d-2)} - \frac{1}{S_d-2}  > 1 - \frac{1}{S_d-2}
\end{align*}
If $g$ is a configuration and hence $\nor{x}_1 \leq 1$, we can bound $x_d$ from above by
\begin{align*}
    x_d \leq 1 - \sum_{i=1}^{d-1} x_i < \frac{1}{S_d-2}
\end{align*}


Recall that the following statements are equivalent:
\begin{itemize}
    \item $\{ Kg \} \in G(\Pi)$ is a configuration i.e. $Kg \in \Ptop$
    \item $Level(Kx) = 1$
\end{itemize}

\begin{lemma} \label{lem-longrun}
    If $g$ is a configuration and $g$ has the long run property, then $(S_d-2) g$ is unique for $d \geq 3$.
\end{lemma}
\begin{proof}
    Let $x \in [0,1]^{d+1}$ such that $x_0 0 + x_1 B_1 + \ldots + x_d B_d = g$ with $\sum_{i=0}^d x_i = 1$. We consider the level $Level(Kx)$ of multiplicities of $x$. Recall that $Level(Kx) = 1$ if and only if $\{ K g \} \in \Pi$ is a configuration. Hence, it remains to prove that that $Level(Kx) > 1$ for every $1 < K \leq S_d-2$.

    By Lemma \ref{lem-level} we know that level $Level(Kx) = Level((K-1)x) - J_K +1$, where $J_K$ is the number of jumps at $K$. This implies by induction that $Level(Kx) = K -J$, where $J$ is the total sum of all jumps in $2x, \ldots , Kx$.  Using that $\sum_{i=1}^{d-1} x_i > 1 - \frac{1}{S_d-2}$, we obtain that $x_0,x_d < \frac{1}{S_d -2}$ and hence $Kx_0, Kx_d < 1$ for $K \leq S_d -2$. This means that component $0$ and component $d$ do not jump in $2x, \ldots , (S_d -2)x$.
    
\begin{observation} For every $1 \leq i \leq d-1$, component $i$ jumps at $1+S_i, 1 + 2S_i,1 + 3S_i, \ldots , 1 + \lfloor \frac{S_d-2}{S_i} \rfloor S_i$. \end{observation}
    Since $(1-\epsilon) \frac{1}{S_i} \leq x_i <  \frac{1}{S_i}$ for $1 \leq i <d$ we know on the one hand that $M S_i x_i < M$ and on the other hand 
    $(1 + M S_i) x_i \geq (1-\epsilon) (\frac{1}{S_i} + M) > M$ as $\epsilon (\frac{1}{S_i} + M) \leq \epsilon (\frac{1+ (S_d-2)}{S_i}) < \frac{S_d-1}{(S_d-1)^2} \frac{1}{S_i} < \frac{1}{S_i}$ for $i \leq d-1$ and $M \leq \frac{S_d-2}{S_i}$. 
    Hence component $i$ jumps at $1+S_i$ from $0$ to $1$ and at $1+2S_i$ from $1$ to $2$ and so on. The total number of jumps $J_K (i)$ in component $i$ can therefore be bounded by $1+ J_K(i) S_i \leq K$ and hence $J_K(i) \leq \lfloor \frac{K-1}{S_i} \rfloor$.

    The total number of jumps $J$ up to $K \leq S_d -2$ in components $0, \ldots , d$  sums up to
    \begin{align*}
        J = \sum_{i=0}^{d} \lfloor \frac{K-1}{S_i} \rfloor = \sum_{i=1}^{d-1} \lfloor \frac{K-1}{S_i} \rfloor \leq \lfloor (K-1) \sum_{i=1}^{d-1} \frac{1}{S_i} \rfloor
    \end{align*}
    Since $\sum_{i=1}^{d-1} \frac{1}{S_i} = 1 - \frac{1}{S_d-1}$ we obtain for $K \leq S_d -2$ 
    \begin{align*}
    J \leq \lfloor (K-1) \sum_{i=1}^{d-1} \frac{1}{S_i} \rfloor = \lfloor (K-1) (1 - \frac{1}{S_d-1}) \rfloor  \leq K-2
    \end{align*}
    which implies that $Level(Kx) = K - J \geq K- (K-2) = 2$ and therefore $\{ Kg \} \not\in \Ptop$ for $K=2 , \ldots , S_d-2$.
    
\end{proof}

\begin{lemma} \label{lem-remainder}
    An element $g \in G(\Pi)$ is a full generator if and only if for all $1 \leq i \leq d$
    \begin{align*}
        g_i \equiv - R_i^{-1} \mod a_i.
    \end{align*}
\end{lemma}
\begin{proof}
Consider the full generator $g$ of a group $G(\PP)$. By definition of the full generator, we obtain that there exists a $z \in \ZZgeq$ such that 
\begin{align*}
Size(g) = \sum_{i=1}^d s_i g_i = \sum_{i=1}^d \frac{g_i}{a_i} = z + \frac{det(B)-1}{det(B)} 
\end{align*}
and hence
\begin{align*}
    det(B)-1 + z \cdot det(B) = det(B) \sum_{i=1}^d \frac{g_i}{a_i} = \sum_{i=1}^d R_i g_i
\end{align*}
By definition of the modulo operation this equation is equivalent to
\begin{align} \label{cong-1}
    \sum_{i=1}^d R_i g_i \equiv det(B) -1 \mod (det(B)).
\end{align}
As $det(B) -1 \equiv -1 \mod a_i$ for each $1 \leq i \leq d$, we obtain by the Chinese remainder Theorem \ref{chinese} (assuming that all $a_i$'s are coprime), that congruence (\ref{cong-1}) is equivalent to the following system of congruences:
\begin{align*}
   \sum_{i=1}^d R_i g_i \equiv - 1 \mod a_i \qquad \text{ for $1 \leq i \leq d$}
\end{align*}
As $R_i \equiv 0 \mod a_j$ for any $i \not= j$, we obtain that $\sum_{i=1}^d R_i g_i \equiv R_j g_j \mod a_j$ and hence
\begin{align*}
        g_i \equiv - R_i^{-1} \mod a_i.
    \end{align*}
\end{proof}
Sylvester's sequence $S_i$ grows double exponentially by approximately $S_i \approx 1.264^{2^{i}}$ and therefore $S_i = 2^{2^{\Omega (i)}}$.
It remains to prove the existence of sizes $s_1, \ldots , s_d$ with group $G(\Pi)$ such that the full generator of $G(\Pi)$ has the long-run property. The following theorem concludes the proof of a double exponential lower bound.
\lowerbound*
\begin{proof}
Given parallelepiped $\Pi = \{ x_1 B_1 + \ldots + x_d B_d \mid x_i \in [0,1) \}$ with configurations $B_i = (0, \ldots ,0 , a_i , 0 , \ldots , 0)^T$. Assume there are sizes $s_i$ such that group $G(\Pi)$ with full generator $g \in \Ptop$ has the long-run property. Then $K$ times configuration $\begin{pmatrix} 1 \\ g \end{pmatrix}$ is by Lemma \ref{lem-unique} the unique representation of the vector $b = \begin{pmatrix} K \\ Kg \end{pmatrix} \in int.cone(\Ptop' \cap \ZZ^d)$ where $\Ptop' = Conv(B'_0, \ldots , B'_d)$ with $B'_0 = (1, 0, \ldots ,0)^T$ and $B'_i = \begin{pmatrix} 1 \\ B_i \end{pmatrix}$.
According to Lemma \ref{lem-longrun} this implies a vertex distance of
$Dist(b) = S_d -2 = 2^{2^{\Omega(d)}}$.
Therefore, it remains to prove the existence of sizes $s_1, \ldots , s_d$ with group $G(\Pi)$ such that the full generator $g$ of $G(\Pi)$ has the long-run property.
In the following we give an inductive construction of the sizes $s_i = \frac{1}{a_i}$:

First, choose $a_1$ arbitrarily such that there is an $m_1$ with $(1-\epsilon) \frac{1}{S_1} \leq \frac{m_1}{a_1} < \frac{1}{S_1}$. This is possible for every $a_1 > \frac{S_1}{\epsilon}$ that is not a multiple of $S_1 = 2$. In this case $m_1$ can be chosen by $m_1 = \lfloor \frac{a_1}{S_1} \rfloor$ and we obtain $\frac{\lfloor a_1 / S_1 \rfloor}{a_1} < \frac{1}{S_1}$ and $\frac{\lfloor a_1 / S_1 \rfloor}{a_1} \geq \frac{(a_1 / S_1) -1}{a_1} \geq \frac{1}{S_1} - \frac{1}{a_1 S_1} \geq (1-\epsilon) S_1$. Additionally, we assume w.l.o.g. that $m_1$ and $a_1$ are coprime.

For $1 \leq i < d$ choose $a_{i+1}$ such that there exists an $m_{i+1}$ with $(1-\epsilon) \frac{1}{S_{i+1}} \leq \frac{m_{i+1}}{a_{i+1}} < \frac{1}{S_{i+1}}$. The existence of the $m_{i+1}$ can be shown for any $a_{i+1} > \frac{S_{i+1}}{\epsilon}$ that is not a multiple of $S_{i+1}$ by the same argument as above for $m_1$. Additionally we choose $a_{i+1}$ such that the following conditions hold:
\begin{align} \label{def-a1}
a_{i+1} &\equiv (\prod_{j=1}^{i-1}a_j)^{-1} \cdot (-m_i)^{-1} \mod a_i
\end{align}
\begin{align} \label{def-a2}
a_{i+1} &\equiv 1 \mod a_j \text{ for $j = 1, \ldots , i-1$}
\end{align}
Remark the following points, where we use the fact that $gcd(a,b) = gcd(a \mod b , b)$ for numbers $a,b \in \ZZ$.
\begin{itemize}
\item The inverse element of $\prod_{j=1}^{i-1}a_j$ and $-m_i$ in $\ZZ / a_i \ZZ$ exists since $a_1, \ldots a_j$ are coprime to $a_i$ and $m_i$ is coprime to $a_i$ (see Theorem \ref{inverse}),
\item since $a_1, \ldots a_i$ are coprime, by the chinese remainder theorem \ref{chinese}, there exists a unique element $a_{i+1} \mod (\prod_{j=1}^i a_j)$ satisfying the above inequalities,
\item condition (\ref{def-a1}) implies that $a_{i+1}$ is coprime to $a_i$ as $m_i$ is coprime to $a_i$ and $\prod_{j=1}^{i-1}a_j$ is coprime to $a_i$ (coprimeness carries over to the inverse),
\item condition (\ref{def-a2}) implies that $a_{i+1}$ is coprime to $a_1, \ldots , a_{i-1}$.
\end{itemize}

\begin{claim}[1]
The full generator $g$ of the constructed group $G(\PP)$ has the long-run property.
\end{claim}
To prove that $g = x_1 B_1 + \ldots + x_d B_d$ has the long-run property, we show for all $1 \leq i < d$ that $\frac{1}{S_i} (1-\epsilon) \leq \frac{g_i}{a_i} < \frac{1}{S_i}$. 
By Lemma \ref{lem-remainder}
\begin{align*}
    g_i \equiv - R_i^{-1} \mod a_i
\end{align*}
By construction of the $a_i$ we obtain for $g_1, \ldots , g_{d-1}$ the following congruences $\mod a_j$:
\begin{align*}
    g_i &\equiv - \left( \prod_{j=1}^{i-1} a_j \cdot \prod_{j=i+1}^{d} a_j \right)^{-1} \stackrel{(\ref{def-a1})}{\equiv} - \left( \prod_{j=1}^{i-1} a_j \cdot a_{i+1} \right)^{-1} \stackrel{(\ref{def-a2})}{\equiv} - \left( (\prod_{j=1}^{i-1} a_j) \cdot (\prod_{j=1}^{i-1} a_j)^{-1} \cdot (-m_i)^{-1} \right)^{-1}
    \\ &\equiv m_i \mod a_i 
\end{align*}
Since for every $\delta \in \Pi$ we have that $\delta_i < a_i$, we know $g_i = m_i$. By definition of $m_i$ we obtain $(1-\epsilon)\frac{1}{S_i} \leq x_i = \frac{m_i}{a_i} < \frac{1}{S_i}$, which proves Claim 1.

\begin{claim}[2]
The full generator $g$ is a configuration.
\end{claim}
Suppose $Level(x) > 1$, then we know by Lemma \ref{lem-size} that $\sum_{i=1}^d x_i = z + \frac{det(B)-1}{det(B)}$ for some $z \in \ZZ_{\geq 1}$.
\begin{align*}
\sum_{i=1}^d x_i < \sum_{i=1}^{d-1} \frac{1}{S_i} + x_d = 
(1 - \frac{1}{S_d -1}) + x_d < (1 - \frac{1}{S_d -1}) + 1
\end{align*}
Since $x_i = \frac{g_i}{a_i} < \frac{1}{S_i}$ for $1 \leq i <d$, we know that $a_i \geq S_i +1$ and hence $det(B) > \prod_{i=1}^{d-1} a_i > \prod_{i=1}^{d-1} (S_i +1) > S_d -1$ which implies:
\begin{align*}
\sum_{i=1}^d x_i < (1 - \frac{1}{det(B)}) + 1 = 1 + \frac{det(B)-1}{det(B)}
\end{align*}
This is a contradiction to $Level(x) > 1$.
\end{proof}

\subsection{Relation between $Dist$ and the IRUP} \label{sec-mirup}
In this section we study briefly the connection between the vertex distance and the modified integer roundup property (MIRUP) which is defined in the following.
Let $\Ptop = \{ x \in \QQgeq^d \mid s^T x \leq 1 \}$ be the knapsack polytope for given sizes $s_1, \ldots s_d \in (0,1]$. For given multiplicities $a_1 , \ldots , a_d $, a packing of the items into a minimum number of bins is given by a solution of the following ILP:
\begin{align} \label{ilpf}
  \min \{ \nor{\lambda}_1 \mid \sum_{p \in \Ptop \cap \ZZ^d} \lambda_p p = \pp , \lambda \in \ZZgeq^d \}.
\end{align}
The relaxed linear program (LP) is defined by
\begin{align} \label{lpf}
  \min \{ \nor{\lambda}_1 \mid \sum_{p \in \Ptop\cap \ZZ^d} \lambda_p p = \pp , \lambda \in \QQgeq^d \}. 
\end{align}
Let $\lambda^*$ be an optimal solution of the ILP (\ref{ilpf}) and let $\lambda^f$ be an optimal solution of the relaxed linear program (\ref{lpf}), then the integrality gap of an instance $(s,b)$ is defined by:
\begin{align*}
    \nor{\lambda}_1 - \nor{\lambda^f}_1
\end{align*}
A well known conjecture by Scheithauer and Terno \cite{Scheithauer199793} concerning the integrality gap for bin packing instance is that for any instance $I$, we have that $\nor{\lambda^*}_1 \leq \lceil \nor{\lambda^f}_1 \rceil +1$ which is the so called modified integer roundup property (MIRUP). The integer roundup property (IRUP) is fulfilled if $\nor{\lambda^*}_1 \leq \lceil \nor{\lambda^f}_1 \rceil$. In general, bin packing instances where die IRUP is not fulfilled appear rarely. In the literature those kind of instances are studied and constructions of instances are given where die IRUP does not hold (see \cite{Scheithauer199793}, \cite{caprara2014}).
In the following we show that a bin packing instance with a large vertex distance $Dist(b)$ implies the existence of many subinstances where the IRUP does not hold. Specifically, we show the following theorem:
\begin{theorem}
Given a bin  packing instance $(s,b)$ corresponding to a vector $b \in int.cone(\Ptop \cap \ZZ^d)$ with vertex distance $Dist(b)$ and let $\lambda \in \ZZ^{\Ptop \cap \ZZ^d}$ be a solution with $\sum_{p \in \Ptop \cap \ZZ^d} \lambda_p p = b$ and vertex distance $\sum_{p \in (\Ptop \cap \ZZ^d) \setminus V_I} \lambda_p = Dist(b)$. For every $\gamma \in (\Ptop \cap \ZZ^d) \setminus V_I$ with $\lambda_\gamma = d+Z$ for some $Z \in \ZZgeq$, there exist at least $Z$ instances where the IRUP does not hold.
\end{theorem}
\begin{proof}
Given an instance $b \in int.cone(\Ptop \cap \ZZ^d)$ with $Dist(b)$. Then there exists an integral optimal solution $\lambda \in \ZZ^{\Ptop \cap \ZZ^d}$ with $\sum_{p \in \Ptop \cap \ZZ^d} \lambda_p p = b$ and $\sum_{p \in (\Ptop \cap \ZZ^d) \setminus V_I} \lambda_p = Dist(b)$. 
We consider for a $\gamma \in (\Ptop \cap \ZZ^d) \setminus V_I$ with $\lambda_\gamma = d+Z$ for some $Z \in \ZZ_{\geq 1}$ the instances $(d+1) \gamma , \ldots , (d+Z) \gamma$. Let $b' \in int.cone(\Ptop \cap \ZZ^d)$ be the vector corresponding to a multiplicity $(d+Z') \gamma$ for a $Z' \leq Z$. Note that by definition of $b'$, we have that $Dist(b') = d+ Z'$, as $\gamma$ is chosen from $(\Ptop \cap \ZZ^d) \setminus V_I$ and the existence of a solution $\lambda''$ for $b'$ with smaller vertex distance would imply a small vertex distance for $b$.
Since $b' \in Cone(\Ptop \cap \ZZ^d)$, there exist a basic feasible solution ${\lambda}^{f} \in \QQgeq^d$ corresponding to vectors $B_1, \ldots , B_d \in \Ptop \cap \ZZ^d$ of LP (\ref{lpf}) with $b' = {\lambda}^{f}_1 B_1 + \ldots+ {\lambda}^{f}_d B_d$ and $\nor{\lambda^f}_1 \leq d+Z'$. Using Caratheodory's theorem, we can assume w.l.o.g. that $B_1, \ldots , B_d$ are vertices.


{\bf Claim}: The vector $[ b' ] = \{ {\lambda}^{f}_1 \} B_1 + \ldots + \{ {\lambda}^{f}_d \} B_d$ does not fulfill the integer roundup property.\\
Suppose the roundup property for $[ b' ]$ is fulfilled, then there exists a packing of instance $[ b' ]$ into $\lceil \nor{ \{{\lambda}^{f} \} }_1 \rceil$ bins. Using the decomposition of $b' = {B \lambda^{f}}^{int} + [ b' ]$ into an integral part ${B \lambda^f}^{int} = \lfloor {\lambda}^{f}_1 \rfloor B_1 + \ldots + \lfloor {\lambda}^{f}_d \rfloor B_d$ and the fractional part $[ b' ]$, we obtain a packing for $b'$ into $\lceil \nor{ {\lambda}^{f} }_1 \rceil \leq d+Z'$ bins (which implies optimality). The  constructed packing has vertex distance of $\leq \lceil \nor{ \{ {\lambda}^{f} \} }_1 \rceil$. Since $\lceil \nor{ \{ {\lambda}^{f} \} }_1 \rceil \leq d < d+Z' = Dist(b')$, this is a contradiction to the minimality of the vertex distance for $b'$.

{\bf Claim}: Let vectors $b^{(1)}, b^{(2)} \in int.cone(\Ptop \cap \ZZ^d)$ be given which correspond to multiplicities $K_1 \gamma$ and $K_2 \gamma$ of vector $\gamma \in (\Ptop \cap \ZZ^d) \setminus V_I$ with $K_1  <  K_2$. Then $[ b^{(1)} ] \not= [ b^{(2)} ]$.\\
By a similar argument as in the previous claim, we can argue in this case. Since $b^{(1)}, b^{(2)} \in Cone(\Ptop \cap \ZZ^d)$, there exist  basic feasible solutions ${\lambda}^{(1)}, {\lambda}^{(2)} \in \QQgeq^d$ corresponding to vectors $B_1, \ldots , B_d \in V_I$ of LP (\ref{lpf}) with $b^{(i)} = {\lambda}^{(i)}_1 B_1 + \ldots+ {\lambda}^{(i)}_d B_d$ for $i=1,2$ and ${\lambda}^{(1)} \leq {\lambda}^{(2)}$ as $K_1 \gamma \leq K_2 \gamma$.
Suppose that $[ {b^{(1)}} ] = [ {b^{(2)}} ]$, then we obtain for the difference $(K_2-K_1) \gamma$ corresponding to ${b^{(2)}}-{b^{(1)}} = {{B} {\lambda^{(2)}}}^{int} + [ {b^{(2)}} ] - {{B} {\lambda^{(1)}}}^{int} - [ {b^{(1)}} ] = {B} {{\lambda^{(2)}}}^{int} - {B}{{\lambda^{(1)}}}^{int}$. Therefore, the difference ${b^{(2)}}-{b^{(1)}}$ can be written by the (positive) sum of vertices $B_1 , \ldots ,B_d$. This implies a packing for ${b^{(2)}}$ by $b^{(2)} = b^{(1)} + (b^{(2)}-b^{(1)})$ with vertex distance $Dist({b^{(1)}}) < Dist({b^{(2)}})$ which contradicts the minimality of vertex distance $b^{(2)}$. 

As a conclusion of the above claims, we obtain for each multiplicity $(d+1) \gamma, \ldots , (d+Z)\gamma$ of $\gamma$ the instances $[(d+1) \gamma], \ldots , [(d+Z)\gamma] \in G(\Pi)$, where die IRUP does not hold.
\end{proof}
Note that an instance with large vertex distance (e.g. double exponential in $d$) implies the existence of solutions with large (double exponential) multiplicities $\lambda_\gamma$ as the number of non-zero components can be bounded by the theorem of Eisenbrand and Shmonin \cite{eisenbrand2006} applied to points in $(\Ptop \cap \ZZ^d) \setminus V_I$.

Together with the construction of the previous subsection where we created a bin packing instance $b$ with a unique solution $\lambda \in int.cone(\Ptop \cap \ZZ^d)$ with $\lambda_\gamma = 2^{2^{\Omega(d)}}$ for some $\gamma \in \Ptop \cap \ZZ^d$, we obtain that $b$ has double exponentially many subinstances where the IRUP is not fulfilled.


\printbibliography

\end{document}